\documentclass{amsart}
\usepackage{amsfonts}
\usepackage{amsmath,amssymb,mathrsfs}

\newtheorem{theorem}{Theorem}
\newtheorem{lemma}[theorem]{Lemma}
\newtheorem{corollary}[theorem]{Corollary}

\theoremstyle{definition}

\newtheorem{remark}{Remark}

\newcommand{\X}{\mathbf X}
\newcommand{\Z}{\mathbb Z}
\newcommand{\N}{\mathbb N}
\newcommand{\Q}{\mathbb Q}
\newcommand{\SE}{\mathcal S}
\newcommand{\PE}{\mathcal P}
\newcommand{\NUL}{\mathcal N}

\DeclareMathOperator{\sgn}{\rm sgn}
\newcommand{\pos}[1]{#1_\oplus}  
\newcommand{\zap}[1]{#1_\ominus}  
\newcommand{\deli}{\bigm|}
\newcommand{\tz}{\bigm|}
\newcommand{\alp}{{\rm alp}}

\newcommand{\ph}{\varphi}
\newcommand{\Xfactor}[1]{\X^{\pos{#1}}-\X^{\zap{#1}}}

\begin{document}
\title{Algebraic properties of word equations}

\author{\v St\v ep\' an Holub}
\author{Jan \v Zemli\v cka}
\address{Department of Algebra, Charles University, Sokolovsk\'a 83, 175 86 Praha, Czech Republic}
\email{holub@karlin.mff.cuni.cz}
\email{zemlicka@karlin.mff.cuni.cz}
\thanks{Supported by the Czech Science Foundation grant number 13-01832S}

\begin{abstract}
The question about maximal size of independent system of word equations is one of the most striking
problems in combinatorics on words. Recently, Aleksi Saarela has introduced a new approach to the problem that is based on linear-algebraic properties of polynomials encoding the equations and their solutions. In this paper we develop further this approach and take into account other algebraic properties of polynomials, namely their factorization. This, in particular, allows to improve the bound for the number of independent equations with maximal rank from quadratic to linear.  
\end{abstract}

\date{\today}

\maketitle

\section{Introduction}

The question about maximal size of independent system of word equations is one of the most striking
problems in combinatorics on words. The conjecture that such a system cannot be infinite (known as Ehrenfeucht's conjecture) had been open for more than a decade, until solved in \cite{Albert} by embedding the free monoid into a metabelian group, and independently in \cite{Guba} by using matrix representation (and generalizing the result to free groups). Both these solutions indicated that, despite strongly combinatorial nature of words, algebraic methods may be necessary to approach problems concerning word equations.

Since the proof of Ehrenfeucht's conjecture, the question about the possible size of independent systems moved to the center of investigation. The above algebraic methods, based ultimately on Hilbert's basis theorem, did not help to approach the problem. The solution is simple for two unknowns because any nontrivial word equation over two unknowns possesses only periodic solutions. However, already for three unknowns, no upper bound is known up to date, leaving open even the possibility that arbitrarily large independent systems of equations over three unknowns exist. This should be contrasted to the fact that the largest known independent system over three unknowns consists of three equations. Some partial results have been obtained in \cite{Dirk} and \cite{Elena}. For general number $n$ of unknowns, independent systems of size $\Theta(n^4)$ have been explicitly constructed in \cite{Juhani}. For more details see also \cite{Lothaire}.

The research got a new impetus recently, when Aleksi Saarela, in \cite{Saarela14}, introduced an idea how to encode word equations to the language of polynomial algebra and exploit its linear algebraic properties. This approach allowed to obtain an upper bound for the size of an independent system of equations over three unknowns, namely quadratic in the length of the shortest equation. In general, the method allows to limit the size of independent systems that have solutions with maximal possible rank.
In this paper we develop further this approach, and applying classical algebraic tools, in particular irreducible factorization of multivariate polynomials, we are able to improve Saarela's bounds from quadratic to linear.  

Basic facts about word equations and corresponding linear algebra are introduced in Section 2. 
The properties of principal solutions (Theorem~\ref{principal}) and their rank are presented here.
The main goal of Section 3 is to translate Saarela's results from less usual notation
of general monoid rings to classical terminology of multivariate rational polynomials which allows to employ
divisibility of obtained polynomials. An application of these tools gives bounds proved in the final section of
this paper. As a consequence of two observations, elementary algebraic (Lemma~\ref{generator})
and finer one of geometrical nature (Theorem~\ref{factors}), we get two pairs of upper bounds.
The first are bounds on the number of pairwise linearly nonequivalent solutions of rank $n-1$ of two strongly independent
equations in $n$ unknowns (Theorem~\ref{main}) and the second are the bounds on the maximal size of strongly independent system
of word equations (Theorem~\ref{main2}).

\section{Solutions of word equations}

In this section we review some facts about systems of word equations and their solutions.

Throughout the paper, $n$ is a number of unknowns, i.e. an integer usually greater then 2, and $\N_0$ denotes the set of all nonnegative integers.
If $\Xi=\{x_1,x_2,\dots,x_n\}$ is a set of unknowns, then a pair $(u,v)\in \Xi^*\times \Xi^*$ is an \emph{equation}. 
 A morphism $h:\Xi^*\to \Sigma^*$ is a \emph{solution} of a system of equations $T=\{(u_i,v_i) \ |\ i\in I\}$ if $h(u_i)=h(v_i)$ for each $i\in I$. If $h(c)$ is the empty word for at least one $c$ in the domain alphabet, then we say that $h$ is \emph{erasing}. The set of letters occurring in a word $w$ is denoted by $\alp(w)$ and $\alp(h)$ denotes $\bigcup_{x\in \Xi} \alp(h(x))$.

The system is \emph{trivial} if $u_i=v_i$ for all $i\in I$.  The vector \[L(h)=(|h(x_1)|,|h(x_2)|,\dots,|h(x_n)|)\in \N_0^n\] is called the \emph{length type} of $h$. 

 For a solution $h$ of the system  $T$, we shall implicitly assume that the domain alphabet $\Xi$ of $h$ is equal to $\alp(T)=\bigcup_{i\in I}\alp(u_iv_i)$. The number of occurrences of a letter $c$ in a word $w$ is denoted by $|w|_c$.

We say that the solution $h'$ of $T$ \emph{divides} the solution $h$ if $h=\vartheta\circ h'$ where $\vartheta$ is not erasing and defined on $\alp(h)$. The solution is called \emph{principle} if it is minimal in the divisibility ordering just defined. In other words, if $h$ is principle and $h=\vartheta\circ h'$ for a solution $h'$, then $\vartheta$ is a renaming of letters.

For sake of completeness we prove the following theorem. Our proof partly follows the standard reference, Proposition 9.5.2 in \cite{lot1}, which, however, is not formulated precisely.
\begin{theorem}\label{principal}
Let  $h:\Xi^*\to \Sigma^*$ be a 
solution of a system $T$. Then there is a unique (up to renaming of letters) principle solution $g$ of $T$ and a unique morphism $\vartheta:\alp(g)^+$ $\to \Sigma^+$ such that $h=\vartheta\circ g$.

Moreover,
\begin{itemize}
	\item $|\alp(g)|<|\alp(T)|$ if $T$ is not trivial; and
	\item  $g$ and $L(\vartheta)$ depend only on $L(h)$ and $T$.
\end{itemize}
\end{theorem}

\begin{proof}
Induction on \[|\alp(h)|+\sum_{x\in\alp(T)}|h(x)|\]
implies that there is at least one principle solution $g$ such that $h=\vartheta\circ g$ for some $\vartheta$.

In order to show that $\vartheta$ is given uniquely and that $g$ and $L(\vartheta)$ are given by $L(h)$, we again proceed by induction. Clearly, $h(x)$ is empty if and only if $g(x)$ is empty. We can therefore suppose that $h$ is not erasing.

If $T$ is trivial, then the only principal solution is identity, $\vartheta=h$ and $L(\vartheta)=L(h)$.

 Let $u\neq v$ for some $(u,v)\in T$, and let $rx$ be a prefix of $u$ and $ry$ a prefix of $v$, where $x\neq y$ are letters, and $r\in \Xi^*$.

Let us first assume that $|h(x)|<|h(y)|$. Then  $h(x)$ is a prefix of $h(y)$. Moreover, also $|g(x)|<|g(y)|$ and $g(x)$ is a prefix of $g(y)$. 
Define $\varphi:\Xi^+\to \Xi^+$ by 
\[
\varphi (z) =
\begin{cases}
xy, & \text{if $z=y$}\\
z, & \text{if $z\neq y$}\,,
\end{cases}
\]
and $h'$, $g'$ by 
\begin{align*}
h'(z)&=
\begin{cases}
h(x)^{-1}h(y) & \text{if $z=y$},\\
h(z) & \text{if $z\neq y$},
\end{cases}
&
g'(z)&=
\begin{cases}
g(x)^{-1}g(y) & \text{if $z=y$},\\
g(z) & \text{if $z\neq y$}.
\end{cases}
\end{align*}
Then  $g'$  and $h'=\vartheta\circ g'$ are solutions of the system $T'=\{(\varphi(u_i),\varphi(v_i))\ |\ i\in I\}$. If $g'=\vartheta'\circ g''$, where $g''$ is a solution of $T'$, then $g''\circ \varphi$ is a solution of $T$ and $g=\vartheta'\circ g''\circ \varphi$, which implies that $\vartheta'$ is a renaming of letters. Therefore $g'$ is a principal solution of $(u',v')$. By induction assumption, $\vartheta$ and $g'$ are unique, and $g'$ and $L(\vartheta)$ are uniquely given by $L(h')$ which is in turn determined by $L(h)$. Since $g=g'\circ \varphi$, we obtain that also $g$ is unique and determined by $L(h)$. Since $\varphi$ is invertible (in the free group), we have $u'\neq v'$ and the induction yields \[|\alp(g)|=|\alp(g')|<|\alp(T')|=|\alp(T)|.\]

The proof is symmetric for $|h(x)|<|h(y)|$. If $|h(x)|=|h(y)|$, then the proof is analogous if we define \[\varphi(y)=x,\quad \text{and} \quad \varphi(z)=z \  \text{otherwise}, \]
and $h'$ and $g'$ are restrictions of $h$ and $g$ respectively on the alphabet of $T'$ which is $\alp(T)\setminus\{y\}$. In this case, the system $T'$ can be trivial but we have $\alp(T')<\alp(T)$.
\end{proof}

Let $h:\{x_1,x_2,\dots,x_n\}^*\to \{a_1,a_2,\dots,a_k\}^*$ be a morphism.
In order to employ linear and polynomial algebra we define non-negative rational vectors $\gamma(h)_i$, $i=1,2,\dots,k$, by
\begin{align*}
	\gamma(h)_i=(|h(x_1)|_{a_i},|h(x_2)|_{a_i},\dots,|h(x_n)|_{a_i}).\label{gamma}
\end{align*}
The set of vectors  $\{\gamma(h)_1,\gamma(h)_2,\dots,\gamma(h)_k\}$ is denoted by $G_h$,
$\Gamma_h$ is the rational vector space generated by $G_h$, and
the dimension of $\Gamma_h$ will be called \emph{rank} of $h$. 
Note that the space $\Gamma_h$ will turn out to be an important linear-algebraic characteristic of a solution $h$ of rank $n-1$, 
that is when $\Gamma_h$ is a hyperplane.

Since rational vectors will serve as one of the main tools in this paper, let us introduce corresponding notation;  
by $\cdot$ we denote the standard dot product on $\Q^n$, for every vector $\alpha\in \Q^n$, 
the $i$-th coordinate of $\alpha$ is denoted by $(\alpha)_i$ and $\pos{\alpha}$, $\zap\alpha$ represents
the uniquely determined nonnegative vectors for which 
$\alpha=\pos\alpha-\zap\alpha$ and $\pos\alpha\cdot\zap\alpha=0$. 
If $M\subseteq \Q^n$, then $M\Q$ is the subspace of the vector space $\Q^n$ generated by $M$.

If $\alpha\in \N_0^r$, then the endomorphism $\vartheta_\alpha$ of $\{a_1,a_2,\dots,a_r\}^*$ will be
defined by the condition $a_i\mapsto a_i^{(\alpha)_i}$. 
Rank of principal solutions has the following property.

\begin{lemma}\label{pismena}
Let $g$ be a principal solution of a system $T$. Then the rank of $g$ is equal to the cardinality of $\alp(g)$.
If $h=\vartheta\circ g$, then the rank of $h$ is at most the rank of $g$.
\end{lemma}
\begin{proof}
We want to show that the set $G_g$ is linearly independent. Suppose it is not. Then there are two distinct vectors $\alpha_1,\alpha_2\in \N^{|\alp(g)|}$ such that $L(\vartheta_{\alpha_1}\circ g)=L(\vartheta_{\alpha_2}\circ g)$. Since both $\vartheta_{\alpha_1}\circ g$ and $\vartheta_{\alpha_2}\circ g$ are solutions of $T$, we obtain a contradiction with Theorem \ref{principal}.

It is easy to see that $G_h\Q$ is a subspace of $G_g\Q$, which implies the second statement.
\end{proof}

\begin{remark}\label{pozn}
There are several kinds of a rank defined in the literature, see for example \cite{handbookcombinatorics}. The ``combinatorial rank'' is the smallest cardinality of a set $A$ such that $h(x)\in A^*$ for each $x\in \Xi$. The combinatorial rank is used in \cite{Saarela14}. 

From the algebraic point of view, most natural is probably the ``free rank'': the size of the basis of the smallest free monoid containing each $h(x)$, $x\in \Xi$. 

We remark without proof that all the ranks, including the ``linear rank'' we use in this paper, coincide for principal solutions. It is therefore  convenient and recommended, whenever possible, to consider, instead of a general solution $h=\vartheta\circ g$, the principal solution $g$ that divides it, see also Lemma \ref{pismena}. The morphism $\vartheta$ typically destroys some properties of the rank.
For example, if it maps two different letters to the same one. As another example, consider a principal solution of rank three and $\vartheta: a\mapsto abc,\ b\mapsto bca,\ c\mapsto cab$. Then $\vartheta$ preserves both combinatorial and free rank but lowers the linear rank to one. 

Note that Theorem \ref{principal} implies that a principal solution of a nontrivial system in $n$ unknowns has rank at most $n-1$, which is known in the combinatorics on words as the  \emph{defect effect}.
\end{remark}

Put $L_g=\{L(\vartheta_\alpha\circ g)\tz \alpha \in \N_0^k\}$.
In order to point out the property of sets $L_g$ which will play an important role in linking word equations and linear algebra,
let us define the notion of \emph{rank} of a subset $M$ of the rational vector space $\Q^n$ as such an integer $r$ that
$\dim M\Q=r$ and $M$ is not covered by any finite union of $(r-1)$-dimensional subspaces of $M\Q$.
Moreover for $M\subseteq \Q^n$ let
$M\N$ denote the set $\left\{\sum_ia_i \alpha_i\tz \ a_i\in\N, \alpha_i\in M \right\}$
and formulate a lemma describing ranks of some subsets of lattice points of $r$-dimensional Euclidean space

\begin{lemma}\label{rank} If $G\subset \Q^n$ and $N\subset \Z^n$ are such that
$G\N\subseteq N\subseteq G\Q$, then $N$ is of rank $\dim G\Q$.
\end{lemma}

\begin{proof} Let $r$ denote $\dim G\Q$. Since $\dim N\Q=r$ by hypotheses, it remains to show that $N$ is not covered by a finite number of spaces with the dimension less than $r$. It is enough to prove that $G\N$ is not covered so.

Consider the set $M=\{(k,k^2,\dots,k^{r}) \tz k\in\N\}\subset\N^r$ and a linearly independent subset $\{\gamma_1,\dots,\gamma_r\}$ of the set $G$.
Each $r$ distinct elements of $M$ are linearly independent in $\Q^r$ since they form a Vandermonde matrix. Therefore also each $r$ distinct elements of the set
\[M_G=\left\{\sum_{i=1}^r k^i\gamma_i \tz k\in \N\right\}\]
are linearly independent. The claim follows since $M_G\subset G\N$ and since any subspace of $\Q^n$ with the dimension less than $r$ contains at most $r-1$ elements of $M_G$.
\end{proof}

The following consequence of the previous lemma relates the rank of a subset of $\Z^n$ with the rank of a morphism.

\begin{lemma}\label{ranks}
If $h:\Xi^*\to \{a_1,a_2,\dots,a_k\}^*$  is a morphism of rank $r$, then the set  $L_h$ is of rank $r$. 
\end{lemma}
\begin{proof} 
Since $G_h\N \subset L_h\subset G_h\Q$, the claim follows from Lemma \ref{rank}.
\end{proof}

The last two lemmas of this section investigate several properties of rank, which
will serve as a useful tool in the next section devoted to polynomial description of world equations.

\begin{lemma}\label{U-rank} If $\mathcal L=\bigcup_{i\le k}\mathcal L_i\subseteq \Q^n$ is a set of rank $r$
such that $\mathcal L_i\N\subseteq\mathcal L_i$ for each $i$, then there exists $i$ such that $\mathcal L_i$ is of rank $r$.
\end{lemma}

\begin{proof}
Since $\mathcal L\subseteq\bigcup_{i\le k}\mathcal L_i\Q$, there exists $i$ such that $\dim \mathcal L_i\Q=r$.
Thus $\mathcal L_i$ contains a linearly independent set $\gamma_1,\dots,\gamma_{r}$. By the hypothesis $\mathcal L_i\N\subseteq\mathcal L_i$,
hence $\mathcal L_i$ is of rank $r$ by Lemma~\ref{rank}.
\end{proof}

If $\lambda\in \Z^n\setminus \{0\}$, let denote the set $\{\alpha\in \Z^n|\ \lambda\cdot\alpha=0\}$ by $\mathcal N(\lambda)$
and put $\mathcal N(\lambda)^+=\mathcal N(\lambda)\cap\N^n$.

\begin{lemma}\label{rank n-1} Let $\lambda\in \Z^n\setminus \{0\}$. 
\begin{enumerate}
\item[(1)] $\mathcal N(\lambda)$ is of rank $n-1$,
\item[(2)] $\mathcal N(\lambda)^+$ is of rank $n-1$ whenever $\pos\lambda\ne 0 \ne \zap\lambda$.
\end{enumerate}
\end{lemma}

\begin{proof} (1)  Since $\mathcal N(\lambda)$ contains the set $G\N$ for every base $G\subseteq\Z^n$
of the $(n-1)$\/-dimensional vector space $\{\mathbf u\in \Q^n|\ \lambda\cdot\mathbf u=0\}$, the assertion follows from Lemma~\ref{rank}.

(2) As $\pos\lambda\ne 0 \ne \zap\lambda$, there exists a positive $\mathbf v$ such that $\lambda\cdot\mathbf v=0$. Hence for every
base $\gamma_1,\dots, \gamma_{n-1}\in \Z^n$
of the vector space $\{\mathbf u\in \Q^n|\ \lambda\cdot\mathbf u=0\}$ there is $c\in\N$ such that 
$G=\{\gamma_1+c\mathbf v,\dots, \gamma_{n-1}+c\mathbf v\}\subset \N^n$ is a linearly independent set.
Since $G\N\subseteq \NUL(\lambda)^+$, it remains to apply Lemma~\ref{rank}.
\end{proof}

\section{Word equations and polynomials}

We review the crucial idea of A. Saarela from \cite{Saarela14} that allows to employ linear algebra.
As it is based on expressing word equations by polynomials, we recall needed notions from polynomial algebra.
$\Q(x)$ denotes the field of fractions of the polynomial ring $\Z[x]$, and $\Q(\X)$ the field of fractions of the polynomial ring $\Z[\X]=\Z[X_1,X_2,\dots,X_n]$.
Let $\gamma\in\N_0^n$. We denote by $\X^\alpha$ the monomial $\prod_{i=1}^nX_i^{(\alpha)_i}\in \Z[\X]$ and
by $\Omega_\gamma:\Z[\X]\to \Z[x]$  the evaluation homomorphism defined by $\Omega_\gamma: X_i\mapsto x^{(\gamma)_i}$, that is, $$\Omega_\gamma\left(p(X_1,\dots,X_n)\right)=p(x^{(\gamma)_1},\dots,x^{(\gamma)_n})\,.$$
In order to simplify the notation, we shall write $p(\gamma)$ instead of $\Omega_\gamma(p)$.

If we choose (a subset of) $\Z$ as the alphabet, then there is a natural representation of a word $w=a_0a_1\cdots a_k$ by the polynomial
$P(w)=\sum_{i=1}^k a_ix^i\in \Z[x]$. In this representation, there is an ambiguity caused by trailing zeros. This can be avoided by using only nonzero digits, or by specifying the length of the word.

Representation of an equation is less obvious. 
Let 
\begin{align}\label{rovnice}
	E= (x_{i_1}x_{i_2}\dots x_{i_r}\,,\,x_{j_1}x_{j_2}\dots x_{j_s})\tag{$*$}
\end{align}
be an equation in $n$ unknowns $\Xi=\{x_1,x_2,\dots,x_n\}$.
Then we define
\[
S_{E,x_j}=\sum_{a: i_a=j}\prod_{t=1}^{a-1}X_{i_t}-\sum_{a: j_a=j}\prod_{t=1}^{a-1}X_{j_t}\in \Z[\X]
\]
(where the empty product is equal to 1).
Comparing this definition with the one given in \cite{Saarela14},
note that, in order to exploit properties of multivariate polynomials,
we work in the usual polynomial ring $\Z[\X]$ instead of the (isomorphic) monoid ring $\Z[X;\mathcal M]$.

If $E$ and $E'$ are two equations, we will be interested in determinants
\[
t_{jk}^{E,E'}:=S_{E,x_j}S_{E',x_k}-S_{E',x_j}S_{E,x_k}\,.
\]

Given a length type $\beta\in \N_0^n$, we obtain the polynomial $S_{E,x_j}(\beta)=\Omega_\beta(S_{E,x_j})\in \Z[x]$.
We define
\begin{align*}
\SE_E&:=\left(S_{E,x_1},S_{E,x_2},\dots,S_{E,x_n}\right)\in  \Z[\X]^n,\\
\SE_E(\beta)&:=\left(S_{E,x_1}(\beta),S_{E,x_2}(\beta),\dots,S_{E,x_n}(\beta)\right)\in  \Z[x]^n\,.
\end{align*}
If $h: \Xi^* \to \Z^*$ is a word homomorphism (that is, $\Z^*$ is understood as a free monoid over the alphabet $\Z$), then denote
\[\PE(h)=\left(P(h(x_1)),P(h(x_2)),\dots,P(h(x_n))\right)\in  \Z[x]^n\,.\]

The point of these definitions is that $h$, with the length type $L=L(h)$, is a solution of $E$ if and only if
\[\SE_E(L)\cdot\PE(h)\,=0,\]
that is, if $\PE(h)$ is a solution of the homogeneous linear equation $\SE_E(L)$. The claim is verified by a straightforward check of definitions.

Trivial cases are described in the following lemma. 
\begin{lemma}\label{nuly}\ 
\begin{enumerate}
	\item $\SE_E=0$ if and only if $E$ is trivial.\label{nulajedna}
	\item If $\SE_E\neq 0$ and $\SE_E(\beta)=0$, then $(\beta)_i=0$ for at least two coordinates $i$.\label{nuladva}
\end{enumerate}
\end{lemma}
\begin{proof} We shall use the notation of \eqref{rovnice}. The proof consists in verifying the following claims directly from definitions.

If $E$ is trivial, then  $\SE_E=0$. Let now $E$ be nontrivial and suppose w.l.o.g. that $r\leq s$. Let $k\geq 1$ be the smallest integer such that ${i_k}\neq j_k$ or, if the left side of $E$ is a prefix of the right side, $k=r+1$.
Then $S_{E,x_{j_k}}\neq 0$. 

If $k=r+1$, then also $S_{E,x_{j_k}}(\beta)\neq 0$ for any $\beta$. If ${i_k}\neq j_k$ and $(\beta)_{i_k}\neq 0$, then $S_{E,x_{j_k}}(\beta)\neq 0$. Similarly, $S_{E,x_{i_k}}(\beta)\neq 0$ if $(\beta)_{j_k}\neq 0$.
\end{proof}

The following lemma is based on the observation that a solution of rank $n-1$ in fact represents $n-1$ linearly independent solutions.

\begin{lemma}\label{repre}
Let $h:\Xi^*\to \{a_1,a_2,\dots,a_k\}^*$ be a solution of rank $n-1$ of equations $E$ and $E'$, and let $\alpha\in \N^k$. Then $\SE_E(L(\vartheta_\alpha\circ h))$ and $\SE_{E'}(L(\vartheta_\alpha\circ h))$ are linearly dependent over $\Q(x)$.
\end{lemma}
\begin{proof}
Let  $h_i=\tau_i\circ \vartheta_\alpha\circ h$, $i=1,2,\dots,k$, where $\tau_i:\alp(h)^*\to \{0,1\}^*$ is defined by $\tau_i(a_i)=\delta_{ij}$. Clearly, each $h_i$ is a solution of both $E$ and $E'$. Observe that the evaluation $x\mapsto 1$ applied to $\PE(h_i)$ yields $(\alpha)_i\gamma(h)_i$. Since such an evaluation yields $n-1$ linearly independent vectors, we deduce that also among $\PE(h_1),\PE(h_2),\dots, \PE(h_{k})$ there are (at least) $n-1$ linearly independent vectors from $\Q(x)^n$.
Since
\[\SE_E(L(\vartheta_\alpha\circ h))\cdot \PE(h_i)=\SE_{E'}(L(\vartheta_\alpha\circ h))\cdot \PE(h_i)=0\]
for each $i=1,2,\dots,n-1$, the proof is completed.
\end{proof}

As we need to know more about polynomials $t_{jk}^{E,E'}$ we formulate several straightforward observations on multivariate polynomials. 

\begin{lemma}\label{evaluation} Let $\alpha,\beta\in \N_0^n$, $c\in \N$ and $\gamma\in \N_0^n$. 
Then:
\begin{enumerate}
\item\label{direct0} $\X^\alpha(\gamma)=x^{\alpha\cdot\gamma}$,
\item\label{direct1}
 $[\X^\alpha-\X^\beta](\gamma)=x^{\beta\cdot\gamma}\left(x^{(\alpha-\beta)\cdot\gamma}-1\right)$ in $\Q(x)$,
\item\label{direct2}
 $[\X^\alpha-\X^\beta](\gamma)=0$ if and only if $(\alpha-\beta)\cdot\gamma=0$,
\item \label{direct3} $\X^{c\alpha}-\X^{c\beta} = (\X^\alpha-\X^\beta)\sum_{i=0}^{c-1}\X^{i\alpha+(c-1-i)\beta}$.
\end{enumerate}
\end{lemma}

\begin{proof} 
(1) It follows immediately from definitions.

(2) The equality is a result of a direct computation:
\[
[\X^{\alpha}-\X^{\beta}](\gamma)=x^{\alpha\cdot\gamma}-x^{\beta\cdot\gamma}=x^{\beta\cdot\gamma}\left(x^{\alpha\cdot\gamma-\beta\cdot\gamma}-1\right)=x^{\beta\cdot\gamma}(x^{(\alpha-\beta)\cdot\gamma}-1).
\]

(3) As $\Q(x)$ is a field, $x^{\beta\cdot\gamma}(x^{(\alpha-\beta)\cdot\gamma}-1)=0$ if and only if $x^{(\alpha-\beta)\cdot\gamma}=1$, which is equivalent to 
$(\alpha-\beta)\cdot\gamma=0$. Now, it remains to apply \eqref{direct1}.

(4) An easy computation.
\end{proof}

We say that $\lambda\in \Z^n$ has {\it coprime coefficients} whenever $\gcd_{i\le n}((\lambda)_i)=1$.

\begin{lemma}\label{div} 
Let $\lambda\in\Z^n\setminus \{0\}$ have coprime coefficients and $\{\gamma_1,\dots, \gamma_{n-1}\}\subseteq \NUL(\lambda)$ be linearly independent
in the rational vector space $\Q^n$. If $\alpha,\beta\in \N_0^n$ such that $[\X^{\alpha}-\X^{\beta}](\gamma_i)=0$ for each $i=1,\dots,n-1$,
then $(\X^{\pos\lambda}-\X^{\zap\lambda})\deli(\X^{\alpha}-\X^{\beta})$.
\end{lemma}

\begin{proof} By Lemma~\ref{evaluation}\eqref{direct2}, the equality $(\alpha-\beta)\cdot\gamma_i=0$ holds for each $i<n$. 
Hence
there exists a nonzero rational number $c$ such that $\alpha-\beta=c\lambda$. Note that $c\in\Z$ because $\lambda$ has coprime coefficients and $\alpha-\beta\in \Z^n$. 
By symmetry, we may suppose without loss of generality that $c\in \N$.
Put 
\[\mu=(\min((\alpha)_1,(\beta)_1),\dots, \min((\alpha)_n,(\beta)_n)).\] It is easy to see that
$c\pos\lambda=\alpha-\mu$ and $c\zap\lambda=\beta-\mu$, thus
\[
\X^{\alpha}-\X^{\beta}= (\X^{\alpha-\mu}-\X^{\beta-\mu})\X^\mu= 
(\X^{c\pos\lambda}-\X^{c\zap\lambda})\X^\mu.
\]
Finally, note that $(\X^{\pos\lambda}-\X^{\zap\lambda})\deli(\X^{c\pos\lambda}-\X^{c\zap\lambda})$
by Lemma~\ref{evaluation}\eqref{direct3}.
\end{proof}

The following lemma is a core observation allowing to employ properties of the unique factorization domain $\Z[\X]$.

\begin{lemma}\label{generator} Let $\lambda\in \Z^n\setminus \{0\}$ have coprime coefficients and let $N\subseteq \NUL(\lambda)$ be of rank $n-1$.
Then 
\begin{itemize}
	\item $p\in \Z[X]$ satisfies $p(\gamma)=0$ for all $\gamma \in N$ if and only if $\left(\X^{\pos\lambda}-\X^{\zap\lambda}\right)\deli p$,
\end{itemize}
and
\begin{itemize}
	\item $\X^{\pos\lambda}-\X^{\zap\lambda}$ is irreducible. 
\end{itemize}
\end{lemma}

\begin{proof}
In the proof we shall often implicitly use the well known fact that the evaluation mapping $\Omega_\alpha$ is a homomorphism.    
By Lemma~\ref{evaluation}(2), we have $[\X^{\pos\lambda}-\X^{\zap\lambda}](\gamma)=0$ for each $\gamma\in\NUL(\lambda)$, hence
$\left(\X^{\pos\lambda}-\X^{\zap\lambda}\right)\deli p$ implies $p(\gamma)=0$ for each $\gamma\in N$.

Assume to the contrary that there is a polynomial $p$ such that $p(\gamma)=0$ for each $\gamma\in N$ and  $p$ is not divisible by $\left(\X^{\pos\lambda}-\X^{\zap\lambda}\right)$. Fix such a $p$ with minimal possible number of monomials. More precisely, since every nonzero coefficient of an arbitrary polynomial in $\Z[\X]$ is a sum of copies of either $1$ or $-1$,
there exist $s,r\in \N$ and two sequences $(\alpha_i\tz \ i\le s)$, $(\beta_i\tz \ i\le r)$ of elements of $\N_0^n$
such that \[p=\sum_{i\le s} \X^{\alpha_i}-\sum_{i\le r} \X^{\beta_i},\] 
 and we suppose that $s+r$ is minimal among all polynomials contradicting the assertion. For each $j\le r$, put
\[
N_j=\{\gamma\in N\tz [\X^{\alpha_1}-\X^{\beta_{j}}](\gamma)=0\}.
\]
Since $p(\gamma)=0$ for each $\gamma\in N$, we deduce that  $N=\bigcup_{j\le r}N_j$ and there exists a linearly independent set $\{\gamma_1,\dots, \gamma_{n-1}\}$ in $N_j$ for some $j$ by Lemma~\ref{U-rank}. Lemma~\ref{div} now yields that $\left(\X^{\pos\lambda}-\X^{\zap\lambda}\right)\deli \left(\X^{\alpha_1}-\X^{\beta_{j}}\right)$, which implies that $p - \left(\X^{\alpha_1}-\X^{\beta_{j}}\right)$ is not divisible by $\left(\X^{\pos\lambda}-\X^{\zap\lambda}\right)$, a contradiction to the minimality of $r+s$.

It remains to prove irreducibility of $\X^{\pos\lambda}-\X^{\zap\lambda}$.
Suppose that $\X^{\pos\lambda}-\X^{\zap\lambda}=gh$ and let 
\begin{align*}
	N_g&=\left\{\gamma \tz g(\gamma)=0 \right\},
&
	N_h&=\left\{\gamma \tz h(\gamma)=0 \right\}.
\end{align*}
Clearly, for each $\gamma\in\NUL(\lambda)$ either $g(\gamma)=0$ or $h(\gamma)=0$. Hence $\mathcal N(\lambda)=N_g \cup N_h$ and at least one of the two sets is of rank $n-1$
by Lemma~\ref{U-rank}. By the first part of the proof, we have that $\X^{\pos\lambda}-\X^{\zap\lambda}$ divides either $g$ or $h$, which we wanted to show.
\end{proof}

An immediate consequence of the last result and Lemma~\ref{evaluation}(4) is that 
the polynomial $\X^{\lambda_0}-\X^{\lambda_1}$ is irreducible 
if and only if $\lambda_0\cdot\lambda_1=0$ and $\lambda_0-\lambda_1$
has coprime coefficients.

Combining the previous results on multivariate polynomials with Lemma \ref{ranks}, we obtain the following observation.

\begin{lemma}\label{dependent}
Let $h$ be a solution of rank $n-1$ of equations $E$ and $E'$.  Then, for each $j,k=1,2,\dots,n$, the determinant $t_{jk}^{E,E'}$ is divisible by $\X^{\pos\lambda}-\X^{\zap\lambda}$ where $\lambda$ has coprime coefficients and $\NUL(\lambda)=\Gamma_h$.
In particular, $\SE_E(\beta)$ and $\SE_{E'}(\beta)$ are linearly dependent over $\Q(x)$ for each $\beta\in \Gamma_h\cap \N_0^{n-1}$.
\end{lemma}
\begin{proof}
Let $t_{jk}=t_{jk}^{E,E'}$. By Lemma \ref{repre},
\[t_{jk}(L(\vartheta_\alpha \circ h))=0\] 
for each $j,k=1,2,\dots,n$ and each $\alpha\in \N^{n-1}$. By Lemma \ref{ranks}, the set $L_h$ is a subset of $\Gamma_h$ of rank $n-1$, and Lemma \ref{generator} implies that $t_{jk}\in \Z[\X]$ is divisible by $\X^{\pos\lambda}-\X^{\zap\lambda}$.
Therefore $t_{jk}(\beta)=0$ for each $\beta\in \Gamma_h\cap \N_0^{n-1}$. This concludes the proof.
\end{proof}

Let $p$ be a polynomial in $\Z[\X]$ such that
\[p=\sum_{i\in I} r_i\]
where $r_i=\pm\X^{\alpha_i}$ with $\alpha_i\in \N_0^n$
and $r_i\neq -r_j$ for $i\neq j$.
We say that the monomial $r_j$, $j\in I$, is \emph{minimal} in $p$, if it has no divisor $r_i$, $i\in I$ with $i\neq j$, that is, if $\alpha_j$ is a minimal element of $\{\alpha_i\ | \ i\in I\}$ 
with respect to the usual product order on $\N_0^n$.

We shall need the following little combinatorial fact.

\begin{lemma}\label{cokolada}
Let $A$ be a set and let $S$ be a subset of $A^k$ such that for each $j\leq k$ there are vectors $s_1,s_2\in A^k$ satisfying
$(s_1)_j\neq (s_2)_j$ and $(s_1)_{j'}= (s_2)_{j'}$ for each $j'\neq j$. Then $S$ has cardinality at least $k+1$.
\end{lemma}
\begin{proof}
Proceed by induction. For $k=1$, the claim is obvious. Let now $k>1$ and consider the set $S'\subset A^{k-1}$ resulting from $S$ by projection on first $k-1$ coordinates. By assumption, the projection is not injective, hence $|S'|<|S|$. It is easy to see that $S'$ satisfies the hypothesis, which implies $|S'|\geq k$, and the proof is complete.   
\end{proof}

\begin{theorem}\label{factors}
Let
\begin{align}\label{faktorizace}
	p=\left(\X^{\lambda_0^{(1)}}-\X^{\lambda_1^{(1)}}\right)\left(\X^{\lambda_0^{(2)}}-\X^{\lambda_1^{(2)}}\right)\cdots \left(\X^{\lambda_0^{(k)}}-\X^{\lambda_1^{(k)}}\right)\sum_{j\in J} e_j\X^{\alpha_j}, \tag{$**$}
\end{align}
where $\left(\X^{\lambda_0^{(i)}}-\X^{\lambda_1^{(i)}}\right)$ are distinct irreducible polynomials, all $\lambda_b^{(i)}$ are non-zero elements of $\N_0^n$, $e_j=\pm 1$  and $\alpha_j\neq \alpha_{j'}$ for any $j,j'\in J$ such that $e_j\neq e_j'$.

Then $p$ contains at least $k+1$ minimal monomials. 
\end{theorem}
\begin{proof}
Let $t\in \Q_+^n$. If $t\cdot \lambda_0^{(i)}\neq t\cdot\lambda_1^{(i)}$ for all $i=1,2,\dots,k$, then we say that $t$ is a \emph{separating type}.
The \emph{profile} of a separating type $t$ is the $k$-tuple $Z(t):=(z_1,\dots,z_k)\in\{1,-1\}^k$, where 
\[z_i:=\sgn\left(t\cdot \lambda_0^{(i)}- t\cdot\lambda_1^{(i)}\right).\] 
 For a separating type $t$ we define $b_{i,t}\in \{0,1\}$, $i=1,2,\dots,k$, $j_t\in J$, and 
\[\rho_t:={\sum \lambda_{b_{i,t}}^{(i)}+\alpha_{j_t}}\]
so that $t\cdot \lambda_{b_{i,t}}^{(i)}< t\cdot \lambda_{1-b_{i,t}}^{(i)}$ for each $i$, and $t\cdot \alpha_{j_t}\leq t\cdot \alpha_{j}$ for all $j\in J$.
 Clearly,  $t\cdot \rho_t\leq t\cdot \beta$ for any monomial $\pm \X^\beta$ resulting from the expansion of \eqref{faktorizace}. Moreover, if $t\cdot \rho_t= t\cdot \beta$, then 
\[\beta={\sum \lambda_{b_{i,t}}^{(i)}+\alpha_{k}}\] for a suitable $k\in J$
and $t\cdot \alpha_{j_t}=t\cdot \alpha_{k}$. Therefore, either $\alpha_{j_t}=\alpha_{k}$, or $\alpha_{j_t}$ and $\alpha_{k}$ are incomparable.
We conclude that for a separating type $t$, the monomial $\X^{\rho_t}$ is minimal.
Also $\rho_{t_1}\neq \rho_{t_2}$ if $t_1$ and $t_2$ are separating types with different profiles, since then $t_1\cdot\rho_{t_1}<t_1\cdot\rho_{t_2}$.

For every $i\le k$ put $\lambda^{(i)}=\lambda_{0}^{(i)}-\lambda_{1}^{(i)}$  and note that $\lambda^{(i)}$ has coprime coefficients, $\lambda_{0}^{(i)}=\pos\lambda^{(i)}$, and
$\lambda_{1}^{(i)}=\zap\lambda^{(i)}$.  
It remains to show that there exist at least $k+1$ separating types with distinct profiles.

First, we show for each $j\le k$ that there exists a vector $t\in \N^n$ such that 
$t\cdot \lambda_{0}^{(j)}=t\cdot \lambda_{1}^{(j)}$, while $t\cdot \lambda_{0}^{(i)}\neq t\cdot \lambda_{1}^{(i)}$ for every $i\ne j$.
Assume to the contrary that for every  $t\in \NUL\left(\lambda^{(j)}\right)^+$ there exists $i\ne j$ such that $t\cdot \lambda_{0}^{(i)}= t\cdot \lambda_{1}^{(i)}$.
Thus $\mathcal N(\lambda^{(j)})^+\subseteq \bigcup_{i\ne j}V_i$ where
$V_i=\left(\mathcal N(\lambda^{(i)})\cap\NUL(\lambda^{(j)})\right)\Q$, $i\ne j$. Since both $\lambda_i$ and $\lambda_j$, $i\ne j$, have coprime coefficients and $\lambda_i\neq \lambda_j$, each $V_i$, $i\ne j$ is an $(n-2)$-dimensional subspace.
By Lemma~\ref{rank n-1}(2), $\mathcal N\left(\lambda^{(j)}\right)^+$ is of rank $n-1$, a contradiction.

Therefore there is a neighborhood of $t$ in $\Q_+^n$ such that $t'\cdot \lambda_{0}^{(i)}\neq t'\cdot \lambda_{1}^{(i)}$, $i\ne j$, for each $t'$ from the neighborhood.
This implies that $(Z(t_+))_j=1$, $(Z(t_-))_j=-1$, and $(Z(t_+))_i=(Z(t_-))_i$, $i\ne j$, for suitable $t_+$ and $t_-$ from the neighborhood.
The proof is completed by Lemma \ref{cokolada}.
\end{proof}

\section{Independent systems of word equations}
In this section, we apply our findings to the question of independence of word equations.

Two systems of equations are \emph{equivalent} if they have the same set of solutions (of all ranks). A system is \emph{independent} if it is not equivalent to any of its proper subsystems. The \emph{compactness property}, proved in \cite{Albert} and \cite{Guba}, states that each system of word equations over $n$ unknowns contains an equivalent finite subsystem. However, very little is known about the possible size of independent systems. There is a lower bound $\Omega(n^4)$ by a direct construction in \cite{defect}, but a nearly complete lack of upper bounds. In fact, results from \cite{Saarela14} discussed here are the only upper bounds known. Note that they depend on the length of equations and apply only to solutions of rank $n-1$. We therefore say that a system  of equations $T$ over $n$ variables is \emph{strongly independent} if any proper subsystem of $T$ has a solution of rank $n-1$ that is not a solution of $T$. 
We can now formulate our goal as to find an upper bound for the size of a strongly independent system.

Let us recall that $\Gamma_h=G_h\Q$ for each morphism $h$.
We say that morphisms $h$ and $h'$ satisfying $\Gamma_h=\Gamma_{h'}$ 
are \emph{linearly equivalent}. 
Note that $h$ and $\vartheta_\alpha\circ h$ are linearly equivalent for each $\alpha\in \N^r$.
Our goal can be achieved by bounding the number of pairwise linearly nonequivalent solutions of rank $n-1$ of two equations.

Let us first look at erasing solutions. Let $\delta_k(E)$ denote the equation in $n-1$ unknowns resulting from $E$ by erasing the variable $x_k$.
\begin{lemma}\label{mazaci}
Let $h$ be an erasing solution of rank $n-1$ of $E$. Then $h(x_k)$ is the empty word for exactly one $k$, and $\delta_k(E)$ is trivial. Moreover $\Gamma_h=\NUL(e_k)$, where $e_k$ is the canonical basis vector defined by $(e_k)_i=\delta_{ki}$.
\end{lemma}
\begin{proof}
The definition of rank implies that $h$ of rank $r$ erases at most $n-r$ variables. Therefore an erasing $h$ of rank $n-1$ erases exactly one variable. The restriction of $h$ on $\Xi\setminus\{x_k\}$ is a solution of rank $n-1$ of an equation $\delta_k(E)$ over $n-1$ variables, therefore $\delta_k(E)$ is trivial by the defect effect (see Remark \ref{pozn}).   
\end{proof}
  We have the following consequence. 
\begin{lemma}\label{dve}
Let $h$ and $h'$ be linearly nonequivalent erasing solutions of rank $n-1$ of equations $E_1$ and $E_2$. Then  $E_1$ and $E_2$ are equivalent. 
\end{lemma}
\begin{proof}
By Lemma \ref{mazaci}, $h$ and $h'$ erase different variables $x_a$ and $x_b$ respectively, and both $E_i=(u_i,v_i)$, $i=1,2$, are of the form 
\begin{align*}
	u_i&=r_0^{(i)}\,x_{j^{(i)}_1}\,r_1^{(i)}\,x_{j^{(i)}_2}\,r_2^{(i)}\,\cdots x_{j^{(i)}_{m_i}}\,r_{m_i}^{(i)}\,,\\
	v_i&=z_0^{(i)}\,x_{j^{(i)}_1}\,z_1^{(i)}\,x_{j^{(i)}_2}\,z_2^{(i)}\,\cdots x_{j^{(i)}_{m_i}}\,z_{m_i}^{(i)}\,
\end{align*}
where $j^{(i)}_k\notin \{a,b\}$, $i=1,2$, $k\leq m_i$, and all $r^{(i)}_k, z^{(i)}_k$ are words in $\{x_a,x_b\}^*$ such that $|r^{(i)}_k|_{x_a}=|z^{(i)}_k|_{x_a}$ and $|r^{(i)}_k|_{x_b}=|z^{(i)}_k|_{x_b}$ for each $r^{(i)}_k,z^{(i)}_k$ with $k\leq m_i$. Therefore both $E_1$ and $E_2$ are equivalent to $x_ax_b=x_b x_a$.
\end{proof}

Note the following fact.
\begin{lemma}
Let $E$ and $E'$ be strongly independent. Then there are $k,l=1,2,\dots,n$ such that $t^{E,E'}_{kl}\neq 0$.
\end{lemma}
\begin{proof}
Independence implies that $E$ and $E'$ are both nontrivial. If $t^{E,E'}_{kl}=0$ for all $k,l=1,2,\dots,n$, then $\SE_E(\beta)$ and $\SE_{E'}(\beta)$ are linearly dependent for each $\beta$. By Lemma \ref{nuly} and Lemma \ref{mazaci}, $\SE_E(L(h))$ and $\SE_{E'}(L(h))$ are both nonzero for any morphism of rank $n-1$. Therefore $\SE_E(L(h))\cdot\PE(h)=0$ if and only if $\SE_{E'}(L(h))\cdot\PE(h)=0$, and $E$ and $E'$ are not strongly independent.
\end{proof}

The next observation is a variant of a similar claim in the proof of \cite[Theorem 5.3]{Saarela14}:

\begin{lemma}\label{minimal} For every pair of equations $E$ and $E'$ and every $k<\ell\le n$, the polynomial $t^{E,E'}_{k\ell}$ contains
at most $2(|E|_{x_k}+|E|_{x_\ell})$ minimal monomials.
\end{lemma}
\begin{proof}
Let 
\begin{align*}
	E&= (x_{i_1}x_{i_2}\dots x_{i_r}\,,\,x_{j_1}x_{j_2}\dots x_{j_s}),&
	E'&= (x_{i'_1}x_{i'_2}\dots x_{i'_{r'}}\,,\,x_{j'_1}x_{j'_2}\dots x_{j'_{s'}}),
\end{align*}
and let
\begin{align*}
	Q_k^+&= \sum_{a: i'_a=k}\prod_{t=1}^{a-1}X_{i'_t}\,,& Q_k^-&= -\sum_{a: j'_a=k}\prod_{t=1}^{a-1}X_{j'_t}\,,\\
	Q_\ell^+&= \sum_{a: i'_a=\ell}\prod_{t=1}^{a-1}X_{i'_t}\,,& Q_\ell^-&= -\sum_{a: j'_a=\ell}\prod_{t=1}^{a-1}X_{j'_t}\,,\\
\end{align*}
so that $S_{E',x_k}=Q_k^+ + Q_k^-$ and $S_{E',x_\ell}=Q_\ell^+ + Q_\ell^-$. Then
\[t_{k\ell}^{E,E'}=S_{E,x_k}Q_\ell^+ + S_{E,x_k}Q_\ell^- +S_{E,x_\ell} Q_k^+ + S_{E,x_\ell}Q_k^-\,.\]
Since monomials in $Q_\ell^+$ are totally ordered by divisibility, there is at most one minimal monomial in $\mu\, Q_\ell^+$, for each monomial $\mu$ in $S_{E,x_k}$. Analogous arguments hold for all four summands in the above expression of $t_{k\ell}^{E,E'}$. Since $S_{E,x_k}$ contains $|E|_{x_k}$ monomials and $S_{E,x_\ell}$ contains $|E|_{x_\ell}$ monomials, the proof is completed.
\end{proof}

We can now prove the desired bounds. 

\begin{theorem}\label{main} Let $E$, $E'$ be strongly independent equations in $n$ unknowns and $m$ be number
of their pairwise linearly nonequivalent solutions of rank $n-1$.
Then 
\begin{itemize}
	\item[(1)]  $m\le |E|+|E'|$, 
	\item[(2)]  there exist indeces $k<\ell\le n$ such that  $m\le 2(|E|_{x_k}+|E|_{x_\ell})$, 
\end{itemize}
\end{theorem}

\begin{proof} 
Let $h_1, h_2,\dots,h_m$ be  pairwise linearly nonequivalent solutions of rank $n-1$ of $E$ and $E'$.
By Lemma \ref{dve}, we can suppose that $h_1,h_2,\dots,h_{m-1}$ are nonerasing.
 For each $i=1,2,\dots,m-1$, let $\lambda^{(i)}$ be a vector with coprime coefficients such that $\Gamma_{h_i}=\NUL\left(\lambda^{(i)}\right)$.
 Let $1\leq k<\ell\leq n$ be such that $t=t^{E,E'}_{k\ell}\ne 0$.

Since $\Xfactor{\lambda^{(i)}}$, $i=1,2,\dots,m$, are pairwise non-associated irreducible polynomials, and  
$(\X^{\pos\lambda}-\X^{\zap\lambda})\deli t$ by Lemma~\ref{dependent}, the product $\prod_{i=1}^m(\Xfactor{\lambda^{(i)}})$ is a divisor of the polynomial $t$.
If $\pos\lambda^{(i)}=0$ or $\zap\lambda^{(i)}=0$, then $\Gamma_{h_i}=\NUL\left(\lambda^{(i)}\right)$
implies that the morphism $h_i$ is erasing. 

(1) As degree of every factor is positive and $\deg(t)\le|E|+|E'|$, the number of pairwise linearly nonequivalent solutions of rank $n-1$ is bounded by $|E|+|E'|$.

(2) Furthermore, all factors $\Xfactor{\lambda^{(i)}}$, $i=1,2\dots,m-1$ satisfy hypotheses of Theorem~\ref{factors}, which implies that  $t$ contains at least $m$ minimal monomials. The proof is completed by Lemma \ref{minimal}.
\end{proof}

As anticipated in \cite[p.16]{Saarela14}, the improvement given by Theorem \ref{main} with respect to \cite[Theorem 5.3.]{Saarela14} has consequences for the size of strongly independent systems of equations. The following lemma shows the main idea (cf. Theorem 3.5 of \cite{Saarela14}).
\begin{lemma}\label{idea}
Let $h$ be a solution of rank $n-1$ of nontrivial equations $E$ and $E'$. Let $h'$ be a solution of $E$ of rank $n-1$ that is not a solution of $E'$. Then $h$ and $h'$ are not linearly equivalent.
\end{lemma}
\begin{proof}
Suppose that $h$ and $h'$ are linearly equivalent. Then $L(h')\in \Gamma_h$, and $\SE_E(L(h'))$ and $\SE_{E'}(L(h'))$ are linearly dependent  by Lemma \ref{dependent}. Since $h'$ is of rank $n-1$, at most one letter can be erased, which implies that both $\SE_E(L(h'))$ and $\SE_{E'}(L(h'))$ are nonzero by Lemma \ref{nuly}. Then $\SE_E(L(h'))\cdot \PE(h')=0$ implies $\SE_{E'}(L(h'))\cdot \PE(h')=0$, a contradiction. 
\end{proof}

\begin{remark}
Let us stress the message of the previous lemma. Independence of equations is defined by distinct sets of solutions. Lemma \ref{idea}, however, shows that the strong independence has more linear algebraic flavor: independent equations are distinguished not only by different solutions $h$ and $h'$ but also by different spaces $\Gamma_h$ and $\Gamma_h'$. More precisely, if two equations have a common solution $h$, then they are equivalent for length types from the whole $\Gamma_h$.

The following example  (suggested by Aleksi Saarela) shows that this is not a vacuous property, since an equation can have (even infinitely many) different but linearly equivalent principal solutions. Consider the well known conjugacy equation $(xz,zy)$. It has infinitely many distinct principal solutions
\begin{align*}
g_i: x\mapsto ab, \quad y\mapsto ba,\quad z\mapsto (ab)^ia,
\end{align*}
but $\Gamma_{g_i}=\NUL((1,-1,0))$ for all $i$.
\end{remark}
   
We are ready to prove the following bounds.
\begin{theorem}\label{main2} 
Let $T=\{E_1,E_2,\dots, E_m\}$ be a strongly independent system. 
Then $m\le |E_1|+|E_2|+2$ and there are $1\leq k<\ell\leq n$ such that $m\leq 2(|E_1|_{x_k}+|E_1|_{x_\ell})+2$. 

If $T$ has a solution of rank $n-1$, then $m\le |E_1|+|E_2|+1$ and there are $1\leq k<\ell\leq n$ such that $m\leq 2(|E_1|_{x_k}+|E_1|_{x_\ell})+1$.
\end{theorem}
\begin{proof}
For each $i=1,\dots,m$, the system $T\setminus\{E_i\}$ has a solution $\ph_i$ of rank $n-1$ 
that is not a solution of $E_i$. 
Let $m\geq 3$ (otherwise the claim is trivial), and let $1\leq i,j,k\leq m$ be three distinct numbers. Then $\ph_i$ is a solution of $E_j$ and $E_k$, while $\ph_j$ is a solution of $E_k$ but not of $E_j$. Lemma \ref{idea} implies that $\ph_i$ and $\ph_j$ are not linearly equivalent.
Therefore $E_1$ and $E_2$ have $m-2$ common solutions $\ph_i$, $i=3,4,\dots,m$, of rank $n-1$, pairwise linearly nonequivalent. The first two bounds now follows from Theorem \ref{main}.

Let $T$ have a solution $\ph_0$. Lemma \ref{idea} again implies that $\ph_0$ is not linearly equivalent to any of $\ph_i$, $i=1,2,\dots,m$, and we have the second claim.
\end{proof}

Recall that an equation $(u,v)$ is \emph{balanced}, if $|u|_x=|v|_x$ for each $x\in \Xi$. If $E$ is not balanced and $h$ is a solution of rank $n-1$, then $\Gamma_h$ is uniquely determined by the length constraints induced by the equation. This implies that strongly independent systems contain balanced equations only. This was first proved in \cite{Dirk} for equations in three unknowns. In \cite{Saarela14}, the result was reproved and generalized to the form presented here.

Our approach and notation allows to characterize balanced equations by the following simple formula.
\begin{lemma}\label{nula}
$E$ be a balanced equation in $n$ unknowns if and only if
\[\left(S_{E,x_1},S_{E,x_2},\dots,S_{E,x_n}\right)\cdot \left(X_1-1,X_2-1,\dots,X_n-1\right)=0. \]
\end{lemma}
\begin{proof}
Let 
\[
E= \left(x_{i_1}x_{i_2}\dots x_{i_r}\,,\,x_{j_1}x_{j_2}\dots x_{j_s}\right).
\]
Then
\begin{align*}
	S_{E,x_\ell}&(X_\ell-1)=\left(\sum_{a: i_a=\ell}\prod_{t=1}^{a-1}X_{i_t}-\sum_{a: j_a=\ell}\prod_{t=1}^{a-1}X_{j_t}\right) (X_\ell-1)=\\
												&=	\left(\sum_{a: i_a=\ell}\prod_{t=1}^{a}X_{i_t}-\sum_{a: j_a=\ell}\prod_{t=1}^{a}X_{j_t}\right)
												- \left(\sum_{a: i_a=\ell}\prod_{t=1}^{a-1}X_{i_t}-\sum_{a: j_a=\ell}\prod_{t=1}^{a-1}X_{j_t}\right).
\end{align*}
This implies that, for each $a<r$ the monomial $\mu=\prod_{t=1}^{a}X_{i_t}$ vanishes in 
\[\left(S_{E,x_1},S_{E,x_2},\dots,S_{E,x_n}\right)\cdot \left(X_1-1,X_2-1,\dots,X_n-1\right) \]
since $\mu$ is contained in $S_{E,x_k}(X_k-1)$ and $-\mu$ is contained in $S_{E,x_{k+1}}(X_{k+1}-1)$. Similarly, the monomial $\prod_{t=1}^{a}X_{j_t}$ vanishes for each $a<s$. Therefore
\[\left(S_{E,x_1},S_{E,x_2},\dots,S_{E,x_n}\right)\cdot \left(X_1-1,X_2-1,\dots,X_n-1\right)= \prod_{t=1}^{s}X_{i_t}-\prod_{t=1}^{r}X_{j_t},\]
which is zero if and only if the equation is balanced.
\end{proof}

Of special interest is the case of three variables, and that for two reasons. First, because this is the simplest case for which there is no bound known independent of the length of equations. In other words, it is an open question  whether independent systems of equations over three variables can be arbitrarily large. Second, because for $n=3$, solutions of rank $n-1$ are precisely all nontrivial (that is, nonperiodic) solutions.
For three variables, we have the following corollary of Lemma \ref{nula}.

\begin{corollary}\label{bal3} 
Let $E_1$ and $E_2$ be balanced equations in variables $\{x_1,x_2,x_3\}$. Then
\[(t_{23},t_{31},t_{12})=t\left(X_1-1,X_2-1,X_3-1\right)\]
for some polynomial $t\in \Z[X_1,X_2,X_3]$.
\end{corollary}
\begin{proof}
Work in the vector space $\Q(X_1,X_2,X_3)^3$. 
Let $\mathbf v_i=\left(S_{E_i,x_1},S_{E_i,x_2},S_{E_i,x_3}\right)$, $i=1,2$. The claim holds for $t=0$ if $\mathbf v_1$ and $\mathbf v_2$ are linearly dependent.
Otherwise, the cross product $(t_{23},t_{31},t_{12})=\mathbf v_1\times \mathbf v_2$ is equal to $t\left(X_1-1,X_2-1,X_3-1\right)$, $t\in \Q(X_1,X_2,X_3)$, by Lemma \ref{nula}. Since $t_{23},t_{31},t_{12}\in \Z[X_1,X_2,X_3]$, it is easy to see that also $t\in \Z[X_1,X_2,X_3]$.
\end{proof}

Theorem \ref{main} and Corollary \ref{bal3} now yield the following claim (compare with \cite[Corollary 6.4]{Saarela14}).

\begin{corollary} 
Let $E_1, \dots ,E_m$ be an independent system of equations in three unknowns having a 
nonperiodic solution. Then 
\begin{itemize}
	\item[(1)]  $m \leq |E_i|+ |E_j|+ 1$ for any pair of distinct equations $E_i, E_j$,
	\item[(2)]  $m \leq 2(|E_1|_x + |E_1|_y) + 1$ for any pair  $x, y$ of unknowns. 
\end{itemize}
\end{corollary}

We conclude by an example, which shows that our results allow to obtain in a simple way concrete information about particular cases. Consider the following  system of two independent equations in three unknowns $x=x_1$, $y=x_2$, and $z=x_3$: 
\begin{align*}
	E_1&=(xyxz,zxyx),\\
	E_2&=(xyxxz,zxxyx).
\end{align*}
We denote $X=X_1$, $Y=X_2$ and $Z=X_3$ and calculate
\begin{align*}
\left(S_{E_1,x},S_{E_1,y},S_{E_1,z}\right)&=\left(1 + XY - Z - XYZ,\, X - XZ,\,  X^2Y-1 \right),\\
\left(S_{E_2,x},S_{E_2,y},S_{E_2,z}\right)&=\left(1 +XY+X^2Y -Z -ZX - X^2YZ,\, X - X^2 Z,\,  X^3 Y - 1\right)
\end{align*}
and
\[(t_{23},t_{31},t_{12})=X(X^2 Y - Z)(X-1,Y-1,Z-1).\]
The polynomial $t=X(X^2 Y-Z)$ characterizes possible nonperiodic common solutions of $E_1$ and $E_2$. Note that the bound of Theorem \ref{main} comes from the bound on the number of hyperplanes covering length types of solutions of rank $n-1$. In other words, this is the number of factors of the form $\Xfactor\lambda$ dividing the corresponding determinant. In the present example, there is only one such factor, namely $(X^2Y-Z)$. This means that each possible common solution $h$ of $E_1$ and $E_2$ of rank two must satisfy $2|h(x)|+|h(y)|=|h(z)|$. Whether such a solution exists can be checked by standard means.

\bibliographystyle{plain}	
\bibliography{polynomy}		

\end{document}